\documentclass[11pt]{article}
\usepackage{amssymb,url}
\usepackage[letterpaper,hmargin=1in,vmargin=1.25in]{geometry}
\usepackage{graphicx,epsfig,color,wrapfig}

\newtheorem{theorem}{Theorem}[section]
\newtheorem{lemma}[theorem]{Lemma}

\newtheorem{corollary}[theorem]{Corollary}

\newenvironment{proof}[1][Proof]{\begin{trivlist}
\item[\hskip \labelsep {\bfseries #1}]}{\end{trivlist}}

\newcommand{\qed}{\nobreak \ifvmode \relax \else
      \ifdim\lastskip<1.5em \hskip-\lastskip
      \hskip1.5em plus0em minus0.5em \fi \nobreak
      \vrule height0.75em width0.5em depth0.25em\fi}

\title{On the Centroids of Symmetrized Bregman Divergences\\ { ---Extended Abstract---}}

\author{ \begin{minipage}{0.47\textwidth}\begin{center}{\Large Frank Nielsen}\\  Sony Computer Science Laboratories, Inc\\ Fundamental Research Laboratory\\ 3-14-13 Higashi Gotanda\\  141-0022 Shinagawa-Ku\\ Tokyo, Japan\\ \texttt{Frank.Nielsen@acm.org}  \end{center}\end{minipage} 
\and 
\begin{minipage}{0.40\textwidth}\begin{center} {\Large Richard Nock}\\  Universit\'e des Antilles-Guyane\\  CEREGMIA\\  Campus de Schoelcher\\  BP 7209, 97275  Schoelcher\\  Martinique, France \\  \texttt{rnock@martinique.univ-ag.org} \end{center}\end{minipage}}

\date{\vskip 0.5cm 21st November 2007 }

\def\tTheta{ {\tilde\Theta} }
\def\tH{ {\tilde H} }

\def\Tr{\mathrm{Tr}}
\def\R{\mathbb{R}}
\def\C{\mathrm{Cone}}

\def\gradF{{\nabla F}}

\def\gradinvF{{(\nabla F})^{-1}}

\def\P{\mathcal{P}}
\def\det{\mathrm{det}}
\def\KL{\mathrm{KL}}

\def\N{\mathcal{N}}

\def\JD{\mathrm{JD}}

\def\JS{\mathrm{JS}}

\def\pvector#1#2{\left(\begin{array}{c}#1\cr #2\end{array}\right)}

\def\X{\mathcal{X}}
\def\V{\mathcal{V}}

\def\dim{\mathrm{dim}}
\def\AVG{\mathrm{AVG}}
\def\S{\mathcal{S}}

\def\grad{\nabla}

\def\pjurl{\url{http://www.sonycsl.co.jp/person/nielsen/BregmanCentroids/}}

\def\equaldef{ {\stackrel{\mathrm{def}}{=}} }

\begin{document}
\maketitle

\begin{abstract}
In this paper, we generalize the notions of centroids and barycenters to the broad class of information-theoretic distortion measures called Bregman divergences.
Bregman divergences are versatile, and  unify quadratic geometric distances with various statistical entropic measures. 
Because Bregman divergences are typically asymmetric, we consider both the left-sided and right-sided centroids and the symmetrized centroids, and prove that  all three are unique.
We give closed-form solutions for the sided centroids that are generalized means, and design a provably fast and efficient approximation algorithm for the symmetrized centroid based on its exact geometric characterization that requires solely to walk on the geodesic linking  the two sided centroids.
We report on our generic implementation for computing entropic centers of image clusters and entropic centers of multivariate normals, and compare our results with former {\it ad-hoc} methods.
\end{abstract}

\vfill
\vfill
\begin{center}
Additional materials including C++ source codes, videos and Java\texttrademark{} applets available at:\\
\pjurl
\end{center}

\thispagestyle{empty}

\newpage
\setcounter{page}{1}
%

\section{Introduction}

Content-based multimedia retrieval applications with their prominent image retrieval systems (CBIRs) are very popular nowadays with the broad availability of massive digital multimedia libraries.
CBIR systems spurred an intensive line of research for better {\it ad-hoc} feature extractions and effective  yet accurate geometric clustering techniques.
In a typical CBIR system~\cite{waveletcbir-2002}, database images are processed offline during a preprocessing step by various feature extractors computing image characteristics such as color histograms. 
These features are aggregated  into signature vectors that represent handles to    images. 
Then given an online query image, the system first computes its signature, and search for the first, say $h$, best matches in the signature space. This requires to define an appropriate {\it similarity measure} between  pairs of signatures. 
Designing an appropriate distance is  tricky since the signature space is often heterogeneous (ie., cartesian product of feature spaces) and the usual Euclidean distance or $L_p$-norms do not always make sense. 
For example, it is better  to use the information-theoretic relative entropy, known as the Kullback-Leibler divergence, to measure the {\it oriented distance} between image histograms~\cite{waveletcbir-2002}. 
{\it Efficiency} is another key issue of CBIR systems since we do not want to compute the similarity measure (query,image) for each image in the database. 
We rather want to prealably {\it cluster} the signatures efficiently during the preprocessing stage for fast retrieval of the  best matches given query signature points.
A first seminal work by Lloyd in 1957~\cite{lloyd-1982} proposed the $k$-means iterative clustering algorithm.
In short, $k$-means starts by choosing $k$ seeds for cluster centers, associate to each point its ``closest'' cluster ``center,''  update the various cluster centers, and reiterate until either convergence is met or the  difference of the ``loss function''  between any two sucessive iterations goes below a prescribed threshold. 
Lloyd choosed the {\it squared} Euclidean distance since the minimum average intracluster distance yields centroids, the centers of mass of the respective clusters, and further proved that $k$-means {\it monotonically} converges to a {\it local} optima. 
Cluster $C_i$'s center $c_i$ is defined by the minimization problem
$c_i=\arg\min_{c} \sum_{p_j\in C_i} ||cp_j||^2=\frac{1}{|C_i|}\sum_{p_j\in C_i} p_j \equaldef \arg\min_{c}\AVG_{L_2^2}(C_i,c)$, where $|C_i|$ denotes the cardinality of $C_i$.
Half a century later, Banerjee et al.~\cite{j-cbd-2005} showed that the $k$-means algorithm {\it extends to} and {\it only} works  for a broad family of distortion measures called Bregman divergences~\cite{Bregman67}. 
Bregman divergences $D_F$ are parameterized families of distortion measures that are defined by a strictly convex and differentiable generator function $F: \X \rightarrow \mathbb{R}^+$ (with $\dim\ \X=d$) as $D_F(p||q)=F(p)-F(q)-<p-q,\gradF(q)>$,
 where $<\cdot,\cdot>$ 
 denotes the inner product
  ($<p,q>=\sum_{i=1}^d p^{(i)}q^{(i)}=p^Tq$)
   and $\gradF(q)$ the gradient at point $q$ 
(ie., $\gradF(q)=\left[\frac{\partial F(q)}{\partial x^{(1)}}, ...,   \frac{\partial F(q)}{\partial x^{(d)}}\right]$).
Further, Teboulle~\cite{Teboulle07} generalized this Bregman $k$-means  algorithm in 2007 by considering both  { hard} and { soft} {\it center-based} clustering algorithms designed for both Bregman~\cite{Bregman67} and Csisz\'ar $f$-divergences~\cite{f-divergence-1966,Csiszar-1967}. 
The fundamental underlying primitive for these {\it center-based} clustering algorithms is to find the intrinsic {\it best single representative} of a cluster. 
As mentionned above, the centroid of a point set $\P=\{p_1, ..., p_n\}$ is defined as the optimizer of the {\it minimum average distance}: $c=\arg\min_{c}\frac{1}{n}\sum_i d(c,p_i)$. 
For oriented distance functions such as Bregman divergences that are not necessarily symmetric, we  thus distinguish {\it sided} and {\it symmetrized} centroids as follows:
$c_R^F  = \arg\min_{c\in\X} \frac{1}{n}\sum_{i=1}^n D_F(p_i||\fbox{c})$, $c_L^F  = \arg\min_{c\in\X} \frac{1}{n}\sum_{i=1}^n D_F(\fbox{c}||p_i)$, and
$c^F  = \arg\min_{c\in\X} \frac{1}{n}\sum_{i=1}^n \frac{D_F(p_i||\fbox{c})+D_F(\fbox{c}||p_i)}{2}$.
The first right-type and left-type centroids $c_R^F$ and $c_L^F$ are called {\it sided centroids}, and the third type centroid $c^F$ is called the {\it symmetrized} Bregman centroid. 
Except for the class of generalized quadratic distances with generator $F_Q(x)=x^TQx$,  $S_F(p;q)=\frac{D_F(p||q)+D_F(q||p)}{2}$ is {\it not} a Bregman divergence, see~\cite{tr2007}.  
Since the three centroids coincide with the center of mass for symmetric Bregman divergences, we consider in the remainder asymmetric Bregman divergences.
We write for short $\AVG_F(\P||c)=\frac{1}{n}\sum_{i=1}^n D_F(p_i||c)$, $\AVG_F(c||\P)=\frac{1}{n}\sum_{i=1}^n D_F(c||p_i)$ and 
$\AVG_F(c;\P)=\frac{1}{n}\sum_{i=1}^n S_F(c;p_i)$
, so that we get respectively
$c_R^F  = \arg\min_{c\in\X} \AVG_F(\P||c), c_L^F  = \arg\min_{c\in\X} \AVG_F(c||\P)$ and $c^F= \arg\min_{c\in\X} \AVG_F(\P;c)$.
The symmetrized Kullback-Leibler~\cite{speechconcatenation-2001,centroidgaussian-2003} and COSH  centroids~\cite{cosh-1991,cosh-2000} (symmetrized Itakura-Saito divergence obtained for $F(x)=-\log x$, the Burg entropy) are certainly the most famous symmetrized Bregman centroids, widely used in image and sound processing.
These symmetrized centroids play a fundamental role in applications that require to handle symmetric information-theoretic distances.

\subsection{Related work, contributions and paper organization}

Prior work in the literature is sparse and disparate. 
We summarize below main references that will be concisely revisited in  section~\ref{sec:sided} under our notational conventions.
Ben-Tal et al.~\cite{entropicmeans-1989} studied {\it entropic means} as the minimum average optimization for various distortion measures such as the $f$-divergences and Bregman divergences. Their study is limited to the sided left-type (generalized means) centroids.
 Basseville and Cardoso~\cite{basseville-isit95} compared in the 1-page paper the generalized/entropic mean values for two entropy-based classes of divergences: $f$-divergences~\cite{Csiszar-1967} and Jensen-Shannon divergences~\cite{csiszar-1991}.
 The closest recent work to our study is Veldhuis' approximation method~\cite{centroidSKL-2002} for computing the symmetrical Kullback-Leibler centroid.

We summarize our contributions as follows:

\begin{itemize}

\item In section~\ref{sec:sided}, we show that the two sided Bregman centroids $c_R^F$ and $c_L^F$ with respect to Bregman divergence $D_F$ are {\it unique} and easily obtained as  {\it generalized means} for the identity and $\gradF$ functions, respectively. 
We extend  Sibson' s notion of {\it information radius}~\cite{informationradius-1969} for these sided centroids, and show that they are both equal to the $F$-Jensen difference, a generalized Jensen-Shannon divergence~\cite{JensenShannon-1991} also known as Burbea-Rao divergences~\cite{jensendifference-1982}.

\item Section~\ref{sec:symmetrized} proceeds by first showing how to reduce the symmetrized $\min \AVG_F(\P;c)$ optimization problem into a simpler system that depends only on the two  sided  centroids $c_R^F$ and $c_L^F$.
We then geometrically characterize {\it exactly} the symmetrized centroid as the intersection point of the geodesic linking the sided centroids with a new type of divergence bisector: the mixed-type bisector. 
This yields a simple and efficient dichotomic search procedure that provably converges fast to the exact symmetrized Bregman centroid. 

\item  The symmetrized Kullback-Leibler divergence ($J$-divergence) and symmetrized Itakura-Saito divergence (COSH distance) are often used in sound/image applications, where our fast geodesic dichotomic walk algorithm  converging to the unique symmetrized Bregman centroid comes in handy over former complex {\it adhoc} methods~\cite{centroidgaussian-2003,cosh-1991,speechconcatenation-2001,veldhuis-2007,normalexpectationcentroid-2001}.
Section~\ref{sec:applications} considers {\it applications} of the generic geodesic-walk algorithm to two cases:
\begin{itemize} 
\item The symmetrized Kullback-Leibler for probability mass functions represented as $d$-dimensional points lying in the $(d-1)$-dimensional simplex $S^d$.
These discrete distributions are handled as  multinomials of the exponential families~\cite{tr2007} with $d-1$ degrees of freedom. 
We instantiate the generic geodesic-walk algorithm for that setting, show how it compares favorably with the prior convex optimization work of Veldhuis~\cite{centroidSKL-2002,veldhuis-2007}, and validate formally  experimental remarks of Veldhuis.  

\item The symmetrized Kullback-Leibler of multivariate normal distributions. 
We describe the geodesic-walk for this particular {\it mixed-type} exponential family of multivariate normals, and explain the Legendre  mixed-type vector/matrix dual convex conjugates defining the corresponding Bregman divergences. This yields a simple, fast and elegant geometric method compared to the former overly complex method of Myrvoll and Soong~\cite{centroidgaussian-2003} that relies on solving Riccati matrix equations.
\end{itemize}
\end{itemize}

\section{Sided Bregman centroids\label{sec:sided}}

\subsection{Right-type centroid}
We first prove that the right-type centroid $c_R^F$ is {\it independent} of the considered Bregman divergence $D_F$: $c_F(\P)=\bar{p}=\frac{1}{n}\sum_{i=1}^n p_i$ is always the center of mass. 
Although this result is well-known in disguise in information geometry~\cite{Amari-Nagaoka-2000}, it was again recently brought up to the attention of the machine learning community by Banerjee et al.~\cite{j-cbd-2005} who proved that Lloyd's iterative  $k$-means ``centroid'' clustering algorithm~\cite{lloyd-1982} generalizes  to the class of Bregman divergences.
We state the result and give the proof for completeness and familizaring us with notations.

\begin{theorem}
The right-type sided Bregman centroid $c_R^F$ of a set $\P$ of $n$ points $p_1$, ...$p_n$, defined as the minimizer for the average right divergence $c_R^F=\arg\min_{c} \sum_{i=1}^n\frac{1}{n}D_F(p_i||c)=\arg\min_{c} \AVG_F(\P||c)$,  is unique, independent of the selected divergence $D_F$,  and coincides with the center of mass $c_R^F=c_R=\bar p=\frac{1}{n}\sum_{i=1}^n p_i$.
\end{theorem}

\begin{proof}
For a given point $q$, the right-type average divergence is defined  as
 $\AVG_F(\P||q)=\sum_{i=1}^n \frac{1}{n} D_F(p_i||q)$.
Expanding the terms $D_F(p_i||q)$'s using the definition of Bregman divergence, 
we get  $\AVG_F(\P||q)=\sum_{i=1}^n \frac{1}{n} \left(F(p_i)-F(q)-<p_i-q, \gradF(q)>\right)$. 
Subtracting and adding $F(\bar p)$ to the right-hand side yields

\begin{eqnarray*}
\AVG_F(\P,q) & = &  
\left(
\sum_{i=1}^n \frac{1}{n} F(p_i) - F(\bar p)  \right) + \left( F(\bar p)-F(q) - \sum_{i=1}^n \frac{1}{n} <p_i-q,\gradF(q)>
 \right),\\
& = & \left(\sum_{i=1}^n \frac{1}{n} F(p_i)  - F(\bar p) \right) + \left( F(\bar p)-F(q) -\left < \sum_{i=1}^n \frac{1}{n}(p_i-q) , \gradF(q) \right > \right),\\
& =& \left(\frac{1}{n}\sum_{i=1}^n  F(p_i) - F(\bar p) \right) + D_F(\bar p||q).
\end{eqnarray*}

Observe that since $\sum_{i=1}^n \frac{1}{n} F(p_i) - F(\bar p)$ is {\it independent} of $q$, minimizing $\AVG_F(\P||q)$ is equivalent to minimizing $D_F(\bar p||q)$.
Using the fact that Bregman divergences $D_F(p||q)$ are non-negative, $D_F(p||q)\geq 0$, and equal to zero {\it if and only if} $p=q$, we conclude that $c_R^F=\arg\min_{q} \AVG_F(\P||q) = \bar p$, namely the center of mass of the point set. 
The minimization remainder, representing the ``information radius'' (by generalizing the notion introduced by Sibson~\cite{informationradius-1969} for the relative entropy),  is 
$
\JS_F(\P)=\frac{1}{n}\sum_{i=1}^n  F(p_i) - F(\bar p) \geq 0$,
 which bears the name of the $F$-Jensen difference\footnote{In the paper~\cite{jensendifference-1982}, it is used for strictly concave function $H=-F$ on a weight distribution vector $\pi$: $J_{\pi}(p_1, ..., p_n)=H(\sum_{i=1}^n\pi_i p_i)-\sum_{i=1}^n \pi_iH(p_i)$. Here, we consider uniform weighting distribution $\pi=u$ (with $\pi_i=\frac{1}{n}$). }~\cite{jensendifference-1982}.
For $F=-H=x\log x$ the negative Shannon entropy, $J_F$ is known as the Jensen-Shannon divergence~\cite{JensenShannon-1991}: $\JS(\P)=H(\sum_{i=1}^n p_i)-\sum_{i=1}^n \frac{1}{n} H(p_i)$.
The Jensen-Shannon divergence is also known as half of the Jeffreys divergence (JD): $\JS(P;Q)=\frac{1}{2}\JD(P;Q)$, and can be interpreted as the  {\it expected information gain}  when discovering which probability distribution  is drawn from (either $P$ or $Q$). 
The Jensen-Shannon divergence can also be interpreted as the {\it noisy channel capacity} with two inputs giving output distributions $P$ and $Q$~\cite{Cover-2006}. Jensen-Shannon divergences are also useful for providing both lower and upper bounds for Bayes probability of error in decision problems~\cite{JensenShannon-1991}.
\end{proof}

\subsection{Dual divergence and left-type centroid}

Before characterizing the {\it left-type} sided Bregman centroid, we recall the fundamental duality of convex analysis: convex conjugation by  Legendre transformation.
We refer to~\cite{tr2007} for detailed explanations that we concisely summarize here as follows:
Any Bregman generator function $F$ admits a {\it dual} Bregman generator function $G=F^*$ via the Legendre transformation
 $G(y)=\sup_{x\in\X}  \{ <y,x>-F(x)  \}$.
The supremum is reached at the {\it unique} point where the gradient
of $G(x)=<y,x> - F(x)$ vanishes, that is when $y=\gradF(x)$. 
 Writing $\X_F'$ for the {\em gradient space} $\{x'=\gradF (x)|
x\in\X\}$,  the convex conjugate 
$G=F^*$ of $F$ is the function $\X_F'\subset \R^d \rightarrow \R$ defined by
$
F^*(x')=<x,x'> - F(x)$.
It follows from Legendre transformation that {\it any} Bregman divergence $D_F$  admits a {\it dual} Bregman divergence $D_{F^*}$ related to $D_F$ as follows:
 $D_F(p||q)=F(p)+F^*(\gradF(q))-<p,\gradF(q)>=F(p)+F^*(q')-<p,q'>=D_{F^*}(q'||p')$. 
 Using the convex conjugation twice, we get the following (dual) theorem for the left-type Bregman centroid:

\begin{theorem}
The left-type sided Bregman centroid $c_L^F$, defined as the minimizer for the average left divergence $c_L^F=\arg\min_{c\in\X}\AVG_L^F(c||\P)$, is the unique point  $c_L^F\in\X$ such that $c_L^F=\gradinvF(\bar{p'})=\gradinvF(\sum_{i=1}^n \gradF(p_i))$, where $\bar{p'}=c_R^{F^*}({\P_{F}}')$ is the center of mass for the gradient point set ${\P_{F}}'=\{p_i'=\gradF(p_i)\ |\ p_i\in\P\}$. 
\end{theorem}

\begin{proof}
Using the dual Bregman divergence $D_{F^*}$ induced by the convex conjugate $F^*$ of $F$, we observe that the left-type centroid $c_L^F=\arg\min_{c\in\X} \AVG_F(c||\P)$ is 
obtained {\em equivalently} by  minimizing the dual right-type centroid problem on the gradient point set:
 $\arg\min_{c'\in\X}'   \AVG_{F^*}({\P_{F}}'||c')$, where we recall that $p'=\gradF(p)$ and  ${\P_{F}}'=\{\gradF(p_1), ..., \gradF(p_n)\}$ denote the gradient point set.
Thus the left-type Bregman centroid $c_L^F$ is computed as the {\it reciprocal gradient} of the center of mass
of the gradient point set $c_R^{F^*}({\P_{F}}')=\frac{1}{n}\sum_{i=1}^n \gradF(p_i)$ :
$c_L^F=\gradinvF(\sum_{i=1}^n \frac{1}{n}\gradF(p_i)) = \gradinvF(\bar{p'})$.
It follows that the left-type Bregman centroid is {\it unique}.
\end{proof}

Observe that the duality also proves that the information radius for the left-type centroid is the {\em same} $F$-Jensen difference (Jensen-Shannon divergence for the convex entropic function $F$).

\begin{corollary}
The information radius equality $\AVG_F(\P||c_R^F)=\AVG_F(c_L^F||\P)=\JS_F(\P)=\frac{1}{n}\sum_{i=1}^n  F(p_i) - F(\bar p)> 0$ is the  $F$-Jensen-Shannon divergence for the uniform weight distribution.
\end{corollary} 
 
\subsection{Generalized means centers and barycenters}

We show that both sided centroids are generalized  means also called quasi-arithmetic or $f$-means. 
We first recall the basic definition of  generalized means\footnote{Studied  independently in 1930 by Kolmogorov and Nagumo, see~\cite{fmeans-2006}. 
A more detailed account is given in~\cite{Hardy1964}, Chapter 3.} that generalizes the usual arithmetic and geometric means.
For a {\it strictly continuous} and {\it monotonous} function $f$, the {\em generalized mean}~\cite{fmeans-2006} of a sequence $\V$ of $n$ real numbers $V=\{v_1, ..., v_n\}$ is defined as 
$
M(\V;f)= f^{-1}(\frac{1}{n}\sum_{i=1}^n f(v_i))$.
The generalized means include the Pythagoras' arithmetic, geometric, and harmonic means, obtained respectively for  functions $f(x)=x$, $f(x)=\log x$ and $f(x)=\frac{1}{x}$ (see appendix~\ref{appendix:sided}).
Note that since $f$ is injective,  its reciprocal function $f^{-1}$ is properly defined.
Further, since $f$ is monotonous, it is  noticed that the generalized mean is necessarily bounded between the {\it extremal set} elements $\min_i v_i$ and $\max_i v_i$:
 $\min_i x_i\leq M(\V;f)\leq \max_i x_i$. 
 In fact, finding these minimum and maximum set elements can be treated themselves  as a special  generalized power mean, another  generalized mean for $f(x)=x^p$ in the limit case $p\rightarrow\pm\infty$.

 These generalized means highlight a bijection: Bregman divergence $D_F\leftrightarrow \gradF$-means.
 The one-to-one mapping holds because  Bregman generator functions $F$ are strictly convex and differentiable functions  chosen up to an affine term~\cite{tr2007}. 
 This affine invariant property {\it transposes} to generalized means as an offset/scaling invariant property:
$M(\S;f)=M(\S;af+b)\ \forall a\in\R^+_* \mbox{\ and\ }\forall b\in\R$.
Although we have considered centroids for simplicity (ie., uniform weight distribution on the input set $\P$), this approach generalizes straightforwardly to {\it barycenters} defined as solutions of minimum average optimization problems for arbitrary unit weight vector $w$ ($\forall i,\ w_i\geq 0$ with $||w||=1$):

\begin{theorem}
Bregman divergences are in bijection with generalized means.
The right-type barycenter $b_R^F(w)$ is independent of $F$ and computed as the weighted arithmetic mean on the point set, a generalized mean for the identity function: $b_R^F(\P;w)=b_R(\P;w)=M(\P;x;w)$ with $M(\P;f;w)= f^{-1}(\sum_{i=1}^n w_i f(v_i))$.
The left-type Bregman barycenter $b_L^F$ is computed as a  generalized mean on the  point set for the gradient function: $b_L^F(\P)=M(\P;\gradF;w)$. 
The information radius of sided barycenters is 
$\JS_F(\P;w)=\sum_{i=1}^d w_iF(p_i)-F(\sum_{i=1}^d w_ip_i)$.
\end{theorem}

\section{Symmetrized Bregman centroid\label{sec:symmetrized}}

\subsection{Revisiting the optimization problem}

For asymmetric Bregman divergences, the symmetrized Bregman centroid is defined by the following optimization problem
$c^F=\arg\min_{c\in\X} \sum_{i=1}^n \frac{D_F(c||p_i)+D_F(p_i||c)}{2}=\arg\min_{c\in\X} \AVG(\P;c)$.
We  simplify this optimization problem to another {\it constant-size} system relying only the right-type and left-type sided centroids, $c_R^F$ and $c_L^F$, respectively.
This will prove that the symmetrized Bregman centroid is uniquely defined as the zeroing argument of a sided centroid function  by generalizing the approach of Veldhuis~\cite{centroidSKL-2002} that studied  the {\it special case} of the symmetrized discrete Kullback-Leibler divergence, also known as $J$-divergence.

\begin{lemma}
The  symmetrized Bregman centroid $c^F$ is unique and obtained by minimizing $\min_{q\in\X} D_F(c_R^F||q)+D_F(q||c_L^F)$: 
$c^F=\arg\min_{q\in\X} D_F(c_R^F||q)+D_F(q||c_L^F)$.
\end{lemma}

\begin{proof}
We have previously shown  that the right-type average divergence can be rewritten as 
$
\AVG_F(\P||q)=\left (\sum_{i=1}^n \frac{1}{n} F(p_i) - F(\bar p) \right) + D_F(\bar p||q)$.
Using Legendre transformation, we have similarly $\AVG_F(q||\P)=\AVG_{F^*}({\P_F}'||q')=(\sum_{i=1}^n \frac{1}{n} F^*(p_i') - F^*(\bar{p'}) ) + D_{F^*}(\bar{p_F'}||q_F')$.
But $D_{F^*}(\bar{p_F'}||q_F')=D_{F^{**}}(\gradF^*\circ\gradF(q)||{\gradF^*}(\sum_{i=1}^n \gradF(p_i)))=D_F(q||c_L^F)$ 
since $F^{**}=F$,  ${\gradF^*}=\gradF^{-1}$ and $\gradF^*\circ\gradF(q)=q$   from Legendre duality.
Combining these two  sum averages, it comes that minimizing $\arg\min_{c\in\X} \frac{1}{2}\left(\AVG_F(\P||q)+\AVG_F(q||\P)\right)$ boils down to minimizing $\arg\min_{q\in\X} D_F(c_R^F||q)+D_F(q||c_L^F)$, 
after removing all terms independent of $q$.
The solution is unique since the optimization problem $\arg\min_{q\in\X} D_F(c_R^F||q)+D_F(q||c_L^F)$ can be itself rewritten as 
$\arg\min_{q\in\X} D_{F^*}(\gradF(q)||\gradF(c_R^F))+D_F(q||c_L^F)$, where $\gradF(q)$ is monotonous and $D_F(\cdot||\cdot)$ and $D_{F^*}(\cdot||\cdot)$ are both  convex in the first argument (but not necessarily in the second). Therefore the optimization problem is convex and admits a unique solution.

\end{proof}

\subsection{Geometric characterization}

We now characterize the exact geometric location of the symmetrized Bregman centroid by introducing a new type of bisector\footnote{See~\cite{tr2007} for the affine/curved and symmetrized bisectors studied in the context of Bregman Voronoi diagrams.} called the mixed-type bisector:

\begin{theorem}
The  symmetrized Bregman centroid $c^F$ is uniquely defined as the minimizer of $D_F(c_R^F||q)+D_F(q||c_L^F)$.
It is defined geometrically as $c^F=\Gamma_F(c_R^F,c_L^F)\cap M_F(c_R^F,c_L^F)$, where $\Gamma_F(c_R^F,c_L^F)=\{\gradinvF((1-\lambda)\gradF(c_R^F)+\lambda\gradF(c_L^F)) \ |\ \lambda\in[0,1]\}$ is the geodesic linking $c_R^F$ to $c_L^F$, and  $M_F(c_R^F,c_L^F)$ is the mixed-type Bregman bisector: 
$M_F(c_R^F,c_L^F)=
\{ x\in\X\ |\ D_F(c_R^F||x)=D_F(x||c_L^F) \}$.
\end{theorem}

{\noindent \bf Proof.}
First, let us prove by contradiction that $q$ necessarily belongs to the geodesic $\Gamma(c_R^F,c_L^F)$.
\begin{wrapfigure}{r}{5cm}
\centering
\includegraphics[bb=0 0 195 124, width=5cm]{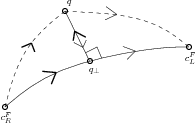}
\caption{\label{fig:perp}The symmetrized Bregman centroid necessarily lies on the geodesic passing through the two sided centroids $c_R^F$ and $c_L^F$.}
\end{wrapfigure}
Assume $q$ does not belong to that geodesic and consider the point $q_\perp$ that is the {\it Bregman perpendicular projection} of $q$ onto the (convex) geodesic~\cite{tr2007}: $q_\perp=\arg\min_{t\in\Gamma(c_R^F,c_L^F)} D_F(t||q)$ as depicted in Figure~\ref{fig:perp}.
Using {\it Bregman Pythagoras' theorem}\footnote{Bregman Pythagoras' theorem is also called the generalized Pythagoras' theorem, and is stated as follows: $D_F(p||q)\geq D(p|| P_{\Omega}(q))+D_F(P_{\Omega}(q)||q)$ where $P_{\Omega}(q)=\arg\min_{\omega\in\Omega} D_F(\omega||q)$ is the Bregman projection of $q$ onto a convex set $\Omega$, see~\cite{j-cbd-2005}. }  twice (see~\cite{tr2007}), we have:
$D_F(c_R^F||q) \geq D_F(c_R||q_\perp)+D_F(q_\perp||q)$ and $D_F(q||c_L^F)\geq D_F(q||q_\perp)+D_F(q_\perp||C_L^F)$.
Thus, we get $D_F(c_R^F||q)+D_F(q||c_L^F) \geq D_F(c_R^F||q_\perp)+ D_F(q_\perp||c_L^F)+ {(D_F(q_\perp||q)+D_F(q||q_\perp))}$.
But since $D_F(q_\perp||q)+D_F(q||q_\perp)>0$, we reach the contradiction since  
$D_F(c_R^F||q_\perp)+D_F(q_\perp||c_L^F)< D_F(c_R^F||q)+D_F(q||c_L^F)$. 
Therefore $q$ necessarily belongs to the geodesic $\Gamma(c_R^F,c_L^F)$.
Second, let us show that $q$ necessarily belongs to the mixed-type bisector.
Assume it is not the case. 
Then $D_F(c_R^F||q)\not = D_F(q||c_L^F)$ and suppose without loss of generality that $D_F(c_R^F||q)> D_F(q||c_L^F)$. 
Let $\Delta=D_F(c_R^F||q)- D_F(q||c_L^F)>0$ and $l_0=D_F(q||c_L^F)$ so that $D_F(c_R^F||q)+D_F(q||c_L^F)=2l_0+\Delta$.
Now move $q$ on the geodesic towards $c_R^F$ by an amount such that $=D_F(q||c_L^F)\leq l_0+\frac{1}{2}\Delta$. 
Clearly, $D_F(c_R^F||q)<l_0$ and   $D_F(c_R^F||q)+D_F(q||c_L^F)<2l_0+\frac{1}{2}\Delta$ contradicting the fact that $q$ was not on the mixed-type bisector.

\def\ttt{3.8cm}
\begin{figure}[th]

\centering

\begin{tabular}{|c|c|} \hline
(a)\includegraphics[bb=0 0 512 512, width=\ttt]{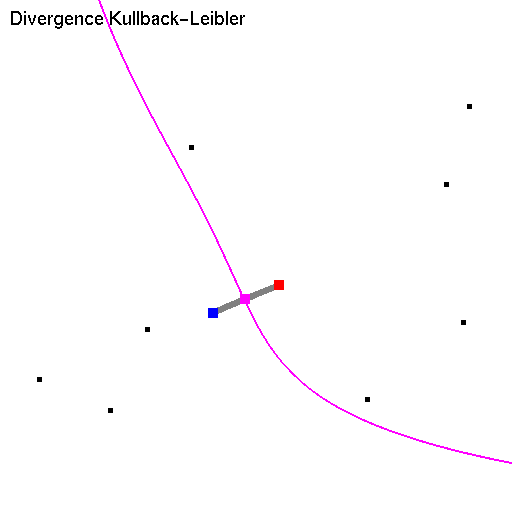} &
(b)\includegraphics[bb=0 0 512 512, width=\ttt]{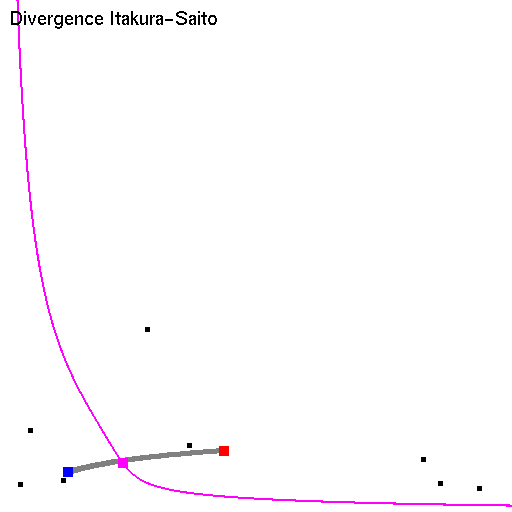} \\  
{\tiny $c_R^F=(0.47,0.78), c_L^F=(0.25,0.76), \AVG_F(\P||c_{R,L}^F)=4.29$} &  {\tiny $c_R^F=(0.43,0.11), c_L^F=(0.13,0.07), \AVG_F(\P||c_{R,L}^F)=22.70$} \\
{\tiny $c^F=(0.35,0.77), , \AVG_F(\P;c^F)=3.96$} & {\tiny $c^F=(0.24,0.09), \AVG_F(\P;c^F)=16.91$} \\ \hline
\end{tabular}

\caption{Bregman centroids for (a) the extended Kullback-Leibler and (b) Itakura-Saito divergences on  the open square $\X=]0,1[^2$.
Right-sided and left-sided, and symmetrized centroids are displayed respectively as red, blue and purple points. 
The geodesic linking the right-sided centroid to the left-sided one is shown in grey, and the mixed-type bisector is displayed in purple.
\label{fig:result}
}

\end{figure}

The equation of the mixed-type bisector $M_F(p,q)$ is neither linear in $x$ nor in $x'=\gradF(x)$ (nor in $\tilde x=(x,x')$) because of the term $F(x)$, 
and can thus only be manipulated implicitly in the remainder: 
$
M_F(p,q)=\{x\in\X\ |\ F(p)-F(q)-2F(x)-<p,x'>+<x,x'>+<x,q'>-<q,q'>=0\}$.
The mixed-type bisector is not necessarily connected (eg., extended Kullback-Leibler divergence),
and  yields the full space $\X$ for symmetric Bregman divergences (ie., generalized quadratic distances).

Using the fact that the symmetrized Bregman centroid necessarily lies on the geodesic linking the two sided centroids $c_R^F$ and $c_L^F$, we get the following corollary:

\begin{corollary}
The symmetrized Bregman divergence minimization problem is both lower and upper bounded as follows:
$ \JS_F(\P) \leq \AVG_F(\P;c^F) \leq D_F(c_R^F||c_L^F)$.
\end{corollary}

Figure~\ref{fig:result} displays the mixed-type bisector, and sided and symmetrized Bregman centroids for the extended\footnote{We relax the probability distributions to belong to the positive orthant $\mathbb{R}_+^d$ (ie., unnormalized probability mass function) instead of the open simplex $\S^d$.} Kullback-Leibler (eKL) and Itakura-Saito (IS) divergences.

\subsection{A simple geodesic-walk dichotomic approximation algorithm}

The exact geometric characterization of the symmetrized Bregman centroid provides us 
a simple method to approximately converge to $c^F$: Namely, we perform a dichotomic walk on the geodesic linking the sided centroids $c_R^F$ and $c_L^F$. 
This dichotomic search yields a novel efficient algorithm that enables us to solve for {\it arbitrary} symmetrized Bregman centroids, beyond the former Kullback-Leibler case\footnote{Veldhuis' method~\cite{centroidSKL-2002}  is based on the general purpose Lagrangian multiplier method with a normalization step. It requires to set up one threshold for the outer loop and two prescribed thresholds for the inner loops. 
For example, Aradilla et al.~\cite{veldhuis-2007} set the number of steps of the outer loop and inner loops  to ten and five iterations each, respectively. Appendix~\ref{appendix:side} provides a synopsis of Veldhuis' method.} of Veldhuis~\cite{centroidSKL-2002}: 
We initially consider $\lambda\in [\lambda_m=0,\lambda_M=1]$ and repeat the following steps until $\lambda_M-\lambda_m\leq\epsilon$, for $\epsilon>0$  a {\it prescribed} precision threshold:

\begin{description}
\item[Geodesic walk.]  Compute interval midpoint $\lambda_h=\frac{\lambda_m+\lambda_M}{2}$ and corresponding geodesic point\\ \centerline{$q_h=\gradinvF((1-\lambda_h)\gradF(c_R^F)+\lambda_h\gradF(c_L^F))$,} 

\item[Mixed-type bisector side.] Evaluate the sign of 
 $D_F(c_R^F||q_h)-D_F(q_h||c_L^R)$, and 
 
\item[Dichotomy.] Branch on $[\lambda_h,\lambda_M]$ if the sign is negative, or on $[\lambda_m,\lambda_h]$ otherwise.
 
\end{description}
  
 Note that {\em any} point on the geodesic (including the midpoint $q_{\frac{1}{2}}$) or on the mixed-type bisector provides an upperbound $\AVG_F(\P;q_h)$ on the minimization task.
 Although it was noted experimentally by Veldhuis~\cite{centroidSKL-2002} for the Kullback-Leibler divergence that this midpoint provides ``experimentally''  a good approximation, let us emphasize that is {\em not true} in  general, as depicted in Figure~\ref{fig:result}(b) for the Itakura-Saito divergence.

\begin{theorem}
The symmetrized Bregman centroid can be approximated within a prescribed precision by a simple dichotomic walk on the geodesic $\Gamma({c_R^F,c_L^F})$ helped by the mixed-type bisector $M_F(c_R^F,c_L^F)$. In general,  symmetrized Bregman centroids do not admit closed-form solutions.
\end{theorem}

In practice, we can control the stopping criterion $\epsilon$ by taking the difference $W_F(q)=D_F(c_R^F||q)-D_F(q||c_L^R)$ between two successive iterations since it monotonically decreases.
The number of iterations can also be theoretically upper-bounded as a function of $\epsilon$ using the maximum value of the Hessian $h_F=\max_{x\in \Gamma(c_R^F,c_L^F)}|| H_F(x)||^2$ along the geodesic $\Gamma(c_R^F,c_L^F)$ by mimicking the analysis in~\cite{nn-ecml-2005} (See Lemma~3 of ~\cite{nn-ecml-2005}).

\section{Applications of the dichotomic geodesic-walk  algorithm\label{sec:applications}}

\subsection{Revisiting the centroid of symmetrized Kullback-Leibler divergence}

Consider a random variable $Q$ on $d$ events $\Omega=\{\Omega_1, ..., \Omega_d\}$, called the sample space.
Its associated discrete distribution $q$ (with $\Pr(Q=\Omega_i)=q^{(i)}$) belongs to the topologically {\it open} $(d-1)$-dimensional probability simplex $\S^d$ of $\mathbb{R}^d_+$: $\sum_{i=1}^d q^{(i)}=1$ and $\forall i\in\{1, ...,d\}\ q_i>0$. Distributions $q$ arise often in practice from  image intensity histograms\footnote{To ensure to all bins of the histograms are non-void, we add a small quantity $\epsilon$ to each bin, and normalize to unit. This is the same as considering the random variable $Q+\epsilon U$ where $U$ is a unit random variable. }. 
To measure the distance between two discrete distributions $p$ and $q$, we use the Kullback-Leibler divergence also known as relative entropy or discrimination information: $
\KL(p||q)=\sum_{i=1}^d p^{(i)}\log \frac{p^{(i)}}{q^{(i)}}$.
Note that this information measure is unbounded whenever there exists $q^{(i)}=0$ for a non-zero $q^{(i)}>0$. But since we assumed that both $p$ and $q$ belongs to the open probability simplex $\S^d$, this case does not occur in our setting: $0\leq \KL(p||q)<\infty$ with left-hand side equality if and only if $p=q$.
The symmetrized KL divergence $\frac{1}{2}(\KL(p||q)+\KL(q||p))$ is also called $J$-divergence or SKL divergence, for short.

The  random variable $Q$ can also be interpreted as a regular exponential family member~\cite{tr2007} in statistics of order $d-1$, generalizing the Bernoulli random variable. Namely, $Q$ is a {\it multinomial} random variable indexed by a $(d-1)$-dimensional {\it parameter vector} $\theta_q$. 
These multinomial distributions belong to the broad class of exponential families~\cite{tr2007} in statistics for which  have the important property that $\KL(p(\theta_p) ||q(\theta_q) )=D_F(\theta_q||\theta_p)$, see~\cite{tr2007}. That is, this property allows us to bypass the fastidious integral computations of Kullback-Leibler divergences and replace it by a simple gradient derivatives for probability distributions belonging to the {\it same} exponential families.
From the canonical decomposition $\exp (<\theta,t(x)>-F(\theta)+C(x))$ of exponential families~\cite{tr2007}, it comes out that
the natural parameters associated with the sufficient statistics $t(x)$ are 
$\theta^{(i)}=\log\frac{ q^{(i)} }{q^{(d)}}=\log\frac{ q^{(i)} }{1-\sum_{j=1}^{d-1} q^{(j)}}$ since $q^{(d)}=1-\sum_{j=1}^{d-1} q^{(j)}$.
The natural parameter space is the topologically open $\mathbb{R}^{d-1}$. 
 The log normalizer is $F(\theta)=\log (1+\sum_{i=1}^{d-1} \exp \theta^{(i)})$, called the multivariate {\it logistic entropy}.
It follows that the gradient is $\gradF(\theta)=\eta=(\eta_i)_i$ with $\eta_i=\frac{\exp \theta^{(i)}}{1+ \sum_{j=1}^{d-1} \exp \theta^{(j)}}$
 and yields the {\it dual parameterization} of the expectation parameters: $\eta=\nabla_\theta F(\theta)$. The expectation parameters play an important role in practice for infering the distributions from identically and independently distributed observations $x_1, ..., x_n$.
Indeed, the maximum likelihood estimator of exponential families is simply given by the center of mass of the  sufficient statistics  computed on the observations: $\hat\eta=\frac{1}{n}\sum_{i=1}^n t(x_i)$, see~\cite{barndorff-nielsen_1988}.
Observe in this case that the log normalizer function  is not separable ($F(x)\not= \sum_{i=1}^{d-1} f_i(x^{(i)})$). 
The function $F$ and $F^*=\int \grad^{-1} F$ are convex conjugates obtained by the Legendre transformation that maps both domains and functions $
(\X_F,F)\longleftrightarrow (\X_{F*},F^{*})
$.
We get the inverse $\grad^{-1} F=(\gradF)^{-1}$ of the gradient $\gradF$ as
$
 \grad^{-1} F(\eta)=\left( \log \frac{\eta^{(i)}}{1-\sum_{j=1}^{d-1} \eta^{(j)}} \right)_i=\theta
$. Thus it comes that the Legendre convex conjugate is
$
F^*(\eta)=\left( \sum_{i=1}^{d-1} \eta^{(i)}\log \eta^{(i)}\right) + (1-\sum_{i=1}^{d-1} \eta^{(i)})\log (1-\sum_{i=1}^{d-1} \eta^{(i)})
$,
the {\it $d$-ary entropy}. 
Observe that for $d=2$, this yields the usual bit entropy\footnote{This  generalizes the 1D case of Kullback-Leibler's Bernoulli divergence:
 $F(x)=\log (1+\exp x)$ is the {\it logistic entropy}, 
$F'(x)=\frac{\exp x}{1+\exp x}$ and ${F'}^{-1}=\log \frac{x}{1-x}$, and
 $F^*(x)=x\log x+(1-x)\log(1-x)$, is the dual {\it bit entropy}.
} function $F^*(\eta)=\eta\log\eta+(1-\eta)\log (1-\eta)$.

To convert back from the multinomial $(d-1)$-order natural parameters $\theta$ to discrete $d$-bin normalized probability mass functions (eg., histograms) $\Lambda\in\S^d$, we use the following mapping:
$
q^{(d)}=\frac{1}{1+\sum_{j=1}^{d-1} (1+\exp\theta^{(j)})}$ and
$
q^{(i)}=\frac{ \exp\theta^{(i)} }{ 1+\sum_{j=1}^{d-1} (1+\exp\theta^{(j)})}$ for all $i\in\{1, ..., d-1\}$.
This gives a {\it valid} ({\it ie.}, normalized) distribution $q\in\S^d$ for {\it any} $\theta\in\R^{d-1}$.
Note that the coefficients in $\theta$ may be either positive or negative depending on the ratio of the probability of the $i$th event with the last one, $q^{(d)}$.

As mentioned above, it turns out that the Kullback-Leibler measure can be computed from the Bregman divergence associated to the multinomial by {\it swapping} arguments: $\KL(p||q)=D_F(\theta_q||\theta_p)$,
where the Bregman divergence $D_F(\theta_q||\theta_p)=F(\theta_q)-F(\theta_p)-<\theta_q-\theta_p,\grad F(\theta_p)>$ is defined
for the strictly convex ($\nabla^2 F>0$) and diffentiable log normalizer $F(\theta)=\log (1+\sum_{i=1}^{d-1} \exp \theta^{(i)})$.
We implemented the geodesic-walk approximation algorithm for that context, and observed in practice that the SKL centroid deviates much (20\% or more in information radius) from the ``middle'' point of the geodesic ($\lambda=\frac{1}{2}$), thus reflecting the asymmetry of the underlying space.
Further, note that our geodesic-walk algorithm {\it proves} the {\it empirical remark} of Veldhuis~\cite{centroidSKL-2002} that ``... the assumption that the SKL centroid is a linear combination of the arithmetic and normalized geometric mean must be rejected.''
Appendix~\ref{appendix:side} displays side by side Veldhuis'  and the geodesic-walk methods for reference, and appendix~\ref{appendix:histogram} report on the sided and symmetrized Bregman centroids of two probability mass functions obtained from intensity histograms of \texttt{apple} images. Observe that the symmetrized centroid distribution {\it may be above} both source distributions, but this is {\it never} the case in the natural parameter domain since the two sided centroids are generalized means, and that the symmetrized centroid belongs to the geodesic linking these two centroids (ie., a barycenter mean of the two sided centroids).

Computing the centroid of a set of image histograms, a center robust to outliers, allows one to design novel applications in information retrieval and image processing.
 For example, we can perform {\it simultaneous contrast} image enhancement by first computing the histogram centroid of a {\it group} of pictures, and then performing histogram normalization to that same reference histogram.

\subsection{Entropic means of multivariate normal distributions}

The probability density function of an arbitary $d$-variate normal $\N(\mu,\Sigma)$ with mean $\mu$ and variance-covariance matrix $\Sigma$    is given by
$
\Pr(X=x)=p(x;\mu,\Sigma)=\frac{1}{(2\pi)^{\frac{d}{2}}\sqrt{\det\Sigma}} \exp \left(-\frac{(x-\mu)^T\Sigma^{-1}(x-\mu)}{2} \right)$.
It is certainly the engineer's favorite family of distributions  that nevertheless becomes intricate to use as dimension goes beyond 3D.
The density function can be rewritten into the canonical decomposition to yield an exponential family of order $D=\frac{d(d+3)}{2}$ (the mean vector and the positive definite matrix  $\Sigma^{-1}$ accounting respectively for $d$ and $\frac{d(d+1)}{2}$ parameters).
The sufficient statistics is {\it stacked} onto a two-part $D$-dimensional vector $\tilde x=(x,-\frac{1}{2}xx^T)$ associated with the natural parameter 
$\tTheta=(\theta,\Theta)=(\Sigma^{-1}\mu,\frac{1}{2}\Sigma^{-1})$. Accordingly, the source parameter are denoted by $\tilde\Lambda=(\mu,\Sigma)$.
The log normalizer specifying the exponential family is 
$F(\tTheta)=\frac{1}{4} \Tr(\Theta^{-1}\theta\theta^T) - \frac{1}{2} \log \det\Theta + \frac{d}{2}\log \pi$ (see~\cite{yoshizawa-tanabe-1999,Amari-Nagaoka-2000}). 
To compute the Kullback-Leibler divergence of two normal distributions $N_p=\N(\mu_p,\Sigma_p)$ and $N_q=\N(\mu_q,\Sigma_q)$, we use the Bregman divergence as follows: $\KL(N_p||N_q)=D_F(\tTheta_q||\tTheta_p)=F(\tTheta_q) - F(\tTheta_p) - <(\tTheta_q-\tTheta_p),\nabla F(\tTheta_p)>$. 
The inner product $<\tTheta_p,\tTheta_q>$ is a {\it composite} inner product obtained as the sum of inner products of vectors and matrices:
$<\tilde \Theta_p,\tilde \Theta_q>=<\Theta_p,\Theta_q>+<\theta_p,\theta_q>$.
For matrices, the inner product $<\Theta_p,\Theta_q>$ is defined by the trace of the matrix product $\Theta_p\Theta_q^T$:  $<\Theta_p,\Theta_q>=\Tr (\Theta_p\Theta_q^T)$.
In this setting, however, computing the gradient, inverse gradient and finding the Legendre convex conjugates are quite involved operations.
 Yoshizawa and Tanabe~\cite{yoshizawa-tanabe-1999} investigated in a unifying framework the differential geometries of the families of probability distributions of {\it arbitrary} multivariate normals from both the viewpoint of Riemannian geometry relying on the corresponding Fisher information metric, and from the viewpoint of Kullback-Leibler information, yielding the classic torsion-free flat shape geometry with dual affine connections~\cite{Amari-Nagaoka-2000}. 
Yoshizawa and Tanabe~\cite{yoshizawa-tanabe-1999} carried out computations that yield the dual natural/expectation coordinate systems arising from the canonical decompotion of the density function $p(x ; \mu,\Sigma)$:

$$
\tilde H=\pvector{\eta=\mu}{H=-(\Sigma+\mu\mu^T)} \Longleftrightarrow  \tilde\Lambda=\pvector{\lambda=\mu}{\Lambda=\Sigma}  \Longleftrightarrow  \tilde\Theta=\pvector{\theta=\Sigma^{-1}\mu}{\Theta=\frac{1}{2}\Sigma^{-1}}
$$

The strictly convex and differentiable dual Bregman generator functions (ie., potential functions in information geometry) are
$F(\tilde \Theta)  =  \frac{1}{4} \Tr(\Theta^{-1}\theta\theta^T)-\frac{1}{2}\log \det \Theta+\frac{d}{2}\log \pi,$ and 
$F^*(\tilde H)  = -\frac{1}{2}\log (1+\eta^TH^{-1}\eta)-\frac{1}{2}\log \det(-H)-\frac{d}{2}\log (2\pi e)$
defined respectively both on the topologically open space $\R^d\times \C_d^-$. Note that removing constant terms does not change the Bregman divergences.
The $\tilde H\Leftrightarrow\tilde \Theta$ coordinate transformations obtained from the Legendre transformation (with $(\nabla F)^{-1}=\nabla F^*$) are given by
$
\tH=\nabla_{\tTheta} F(\tTheta)=\pvector{\nabla_{\tTheta} F(\theta)}{\nabla_{\tTheta} F(\Theta)}  = 
\pvector{\frac{1}{2}\Theta^{-1}\theta}{-\frac{1}{2}\Theta^{-1}-\frac{1}{4}(\Theta^{-1}\theta)(\Theta^{-1}\theta)^T}
=
\pvector{\mu}{-(\Sigma+\mu\mu^T)}$ and 
$
\tTheta=\nabla_{\tH} F^*(\tH)=
\pvector{\nabla_{\tH} F^*(\eta)}{\nabla_{\tH} F^*(H)}  = 
\pvector{-(H+\eta\eta^T)^{-1}\eta}{-\frac{1}{2}(H+\eta\eta^T)^{-1}}=
\pvector{\Sigma^{-1}\mu}{\frac{1}{2}\Sigma^{-1}}
$.
These formula simplifies significantly when we restrict ourselves to diagonal-only variance-covariance matrices $\Sigma_i$, spherical normals $\Sigma_i=\sigma_i I$, or univariate normals $\mathcal{N}(\mu_i,\sigma_i)$.

Computing the symmetrized Kullback-Leibler centroid of a set of normals (Gaussians) is an essential operation for clustering sets of multivariate normal distributions using center-based $k$-means algorithm~\cite{gaussianclustering-2006,Teboulle07}.
Myrvoll and Soong~\cite{centroidgaussian-2003}  described the use of multivariate normal clustering  in automatic speech recognition.
They derived a numerical local algorithm for computing the  multivariate normal centroid by solving iteratively Riccati matrix equations,  initializing the solution to the so-called ``expectation centroid''~\cite{normalexpectationcentroid-2001}.
Their method is a complex and costly since it also involves  solving for eigensystems. 
In comparison, our geometric geodesic dichotomic walk procedure for computing the entropic centroid, a Bregman symmetrized centroid, yields an extremely fast and simple algorithm with {\it guaranteed} performance.

\section*{Acknowledgements.}
We gratefully thank Professors Lev M. Bregman~\cite{Bregman67},  Marc Teboulle~\cite{entropicmeans-1989}, and  Baba Vemuri~\cite{DTISegmentation-2005} for email correspondences and sending us printed copies of seminal papers.
We also thank Guillaume Aradilla for sharing with us his experience~\cite{veldhuis-2007} concerning Veldhuis's algorithm~\cite{centroidSKL-2002}.


\appendix

\section{Dominance relationships of sided centroid coordinates\label{appendix:sided}}

The table below illustrates the bijection between Bregman divergences and generalized $f$-means for the Pythagoras' means (ie., extend to separable Bregman divergences):
  \vskip 0.3cm
{
\renewcommand{\arraystretch}{1.0}
{
\centering
\begin{tabular}{|cccccc|} \hline\hline
Bregman divergence $D_F$ & $F$ & $\longleftrightarrow$ &  $f=F'$ & $f^{-1}=(F')^{-1}$ &  $f$-mean  \\
& & & & & (Generalized means) \\ \hline\hline
Squared Euclidean distance  & $\frac{1}{2}x^2$ &  $\longleftrightarrow$ & $x$& $x$ & Arithmetic mean  \\ 
 (half squared loss) & & & & & $\sum_{j=1}^n \frac{1}{n} x_j$ \\ \hline
Kullback-Leibler divergence & $x\log x-x$ &  $\longleftrightarrow$&  $\log x$  &  $\exp x$ &  Geometric mean\\ 
  (Ext. neg. Shannon entropy) &  & & & & $( \prod_{j=1}^n x_j )^{\frac{1}{n}}$   \\ \hline
Itakura-Saito divergence & $-\log x$ & $\longleftrightarrow$ & $-\frac{1}{x}$ &  $-\frac{1}{x}$ & Harmonic mean   \\
  (Burg entropy) & & & & & $\frac{n}{\sum_{j=1}^n \frac{1}{x_j}}$   \\ 
 \hline\hline
\end{tabular}\vskip 0.2cm
}

}

We  give a characterization of the coordinates ${c_R^F}^{(i)}$ of the right-type average centroid (center of mass) with respect to those of the left-type average centroid, the ${c_L^F}^{(i)}$ coordinates.

{\noindent\bf Corollary}\\
Provided that $\gradF$ is convex (e.g., Kullback-Leibler divergence), we have
 ${c_R^F}^{(i)}\geq {c_L^F}^{(i)}$ for all $i\in\{1, ..., d\}$.
Similarly, for concave gradient function (e.g., exponential loss), we have  ${c_R^F}^{(i)}\leq {c_L^F}^{(i)}$ for all $i\in\{1, ..., d\}$.

\begin{proof}
Assume $\gradF$ is convex and apply Jensen's inequality to $\frac{1}{n}\sum_{i=1}^n \gradF(p_i)$.
Consider for simplicity without loss of generality 1D functions. We have  
$
\frac{1}{n}\sum_{i=1}^n \gradF(p_i) \leq \gradF(\frac{1}{n}\sum_{i=1}^n p_i )
$. Because $\gradinvF$ is a monotonous function, we get $c_L^F=\gradinvF (\frac{1}{n}\sum_{i=1}^n \gradF(p_i))\leq \gradinvF(\gradF(\frac{1}{n}\sum_{i=1}^n p_i ))=\frac{1}{n}\sum_{i=1}^n p_i=c_R^F$. 
Thus we conclude that ${c_R^F}^{(i)}\geq {c_L^F}^{(i)}\ \forall i\in\{1, ..., d\}$ for convex $\gradF$ (proof performed coordinatewise).
For concave $\gradF$ functions (i.e., dual divergences of $\gradF$-convex primal divergences), we simply reverse the inequality (e.g., the exponential loss dual of the Kullback-Leibler divergence). 
\end{proof}

Note that Bregman divergences $D_F$ may neither have  their gradient $\gradF$ convex nor concave.
The bit entropy $F(x)=x\log x+(1-x)\log (1-x)$  yielding the logistic loss $D_F$ is such an example.
In that case, we cannot {\it a priori} order  the coordinates of $c_R^F$ and $c_L^F$.

This dominance relationship can be verified for the plot in natural parameter space of Appendix~C.


\section{Synopsis of Veldhuis' and the generic geodesic-walk methods\label{appendix:side}}

The table below provides a side-by-side comparison of Veldhuis' $J$-divergence centroid convex programming method~\cite{centroidSKL-2002} with our generic symmetrized Bregman centroid (entropic means) geodesic-walk instantiated for the Kullback-Leibler divergence. 
\vskip 0.3cm

\begin{tabular}{|c|c|}\hline
Veldhuis' algorithm & Geodesic-walk algorithm\\ \hline
\begin{minipage}[c]{7.5cm}
\includegraphics[bb=0 0 698 1025, width=7.5cm]{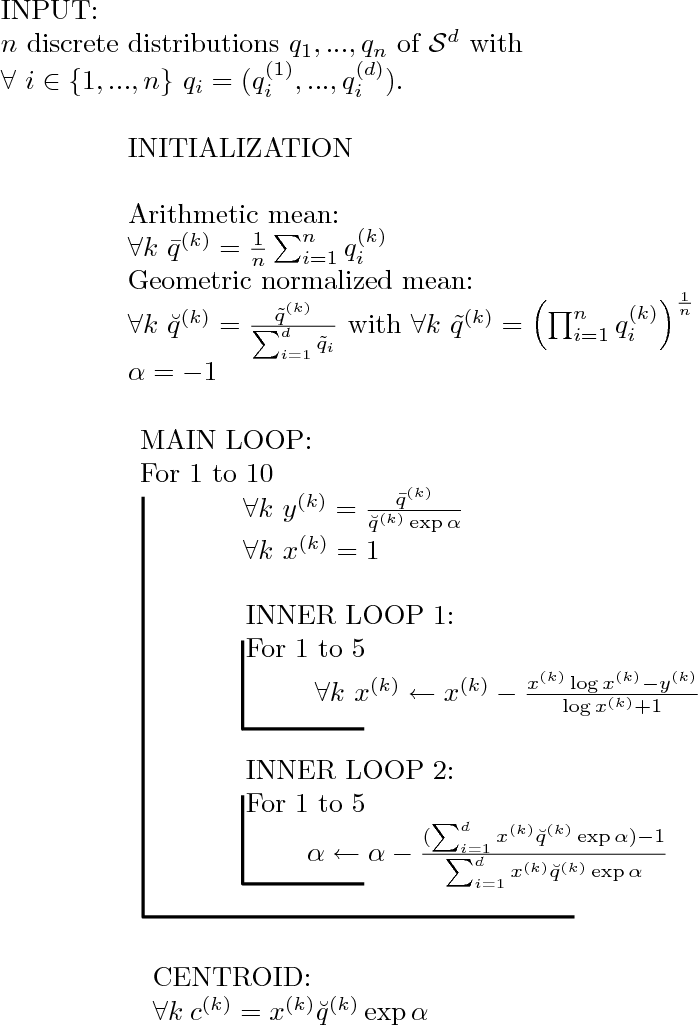} 
\end{minipage}
&
\begin{minipage}[c]{7.5cm}
\includegraphics[bb=0 0 538 941, width=7.5cm]{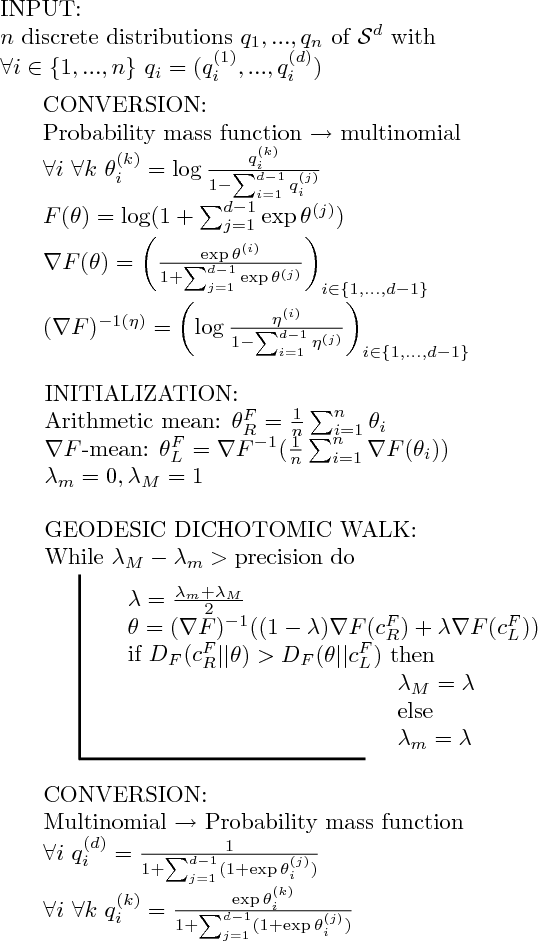}
\end{minipage}\\ \hline
\end{tabular}
\vskip 0.3cm
Both C++ source codes with cross-check validations are available at\\
\centerline{\pjurl}

\newpage

\section{Image histogram centroids with respect to the relative entropy\label{appendix:histogram}}
The plots below show the Kullback-Leibler sided and symmetrized centroids on two distributions taken as the intensity histograms of  the \texttt{apple} images shown below. Observe that the symmetrized centroid distribution is {\it above} both source distributions for intensity range $[100-145]$, but this is {\it never} the case in the natural parameter space due to the property of generalized means. 

\vskip 0.3cm

{
\begin{center}
\includegraphics[bb=0 0 257 250, width=2.5cm, height=3cm]{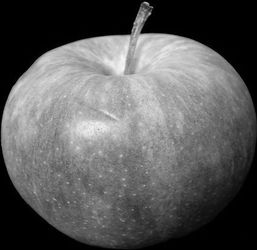}\hskip 2.5cm
\includegraphics[bb=0 0 257 250, width=2.5cm, height=3cm]{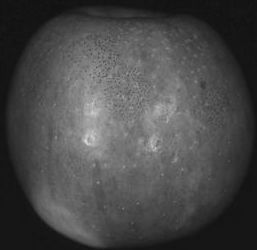}\hskip 2.5cm
\includegraphics[bb=0 0 625 360, width=4cm ]{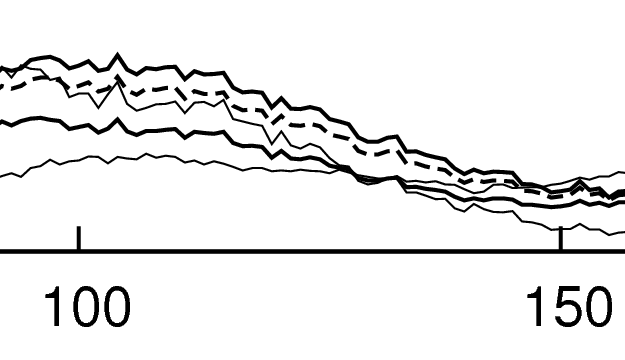}\\
\includegraphics[bb=0 0 1343 932, width=10cm]{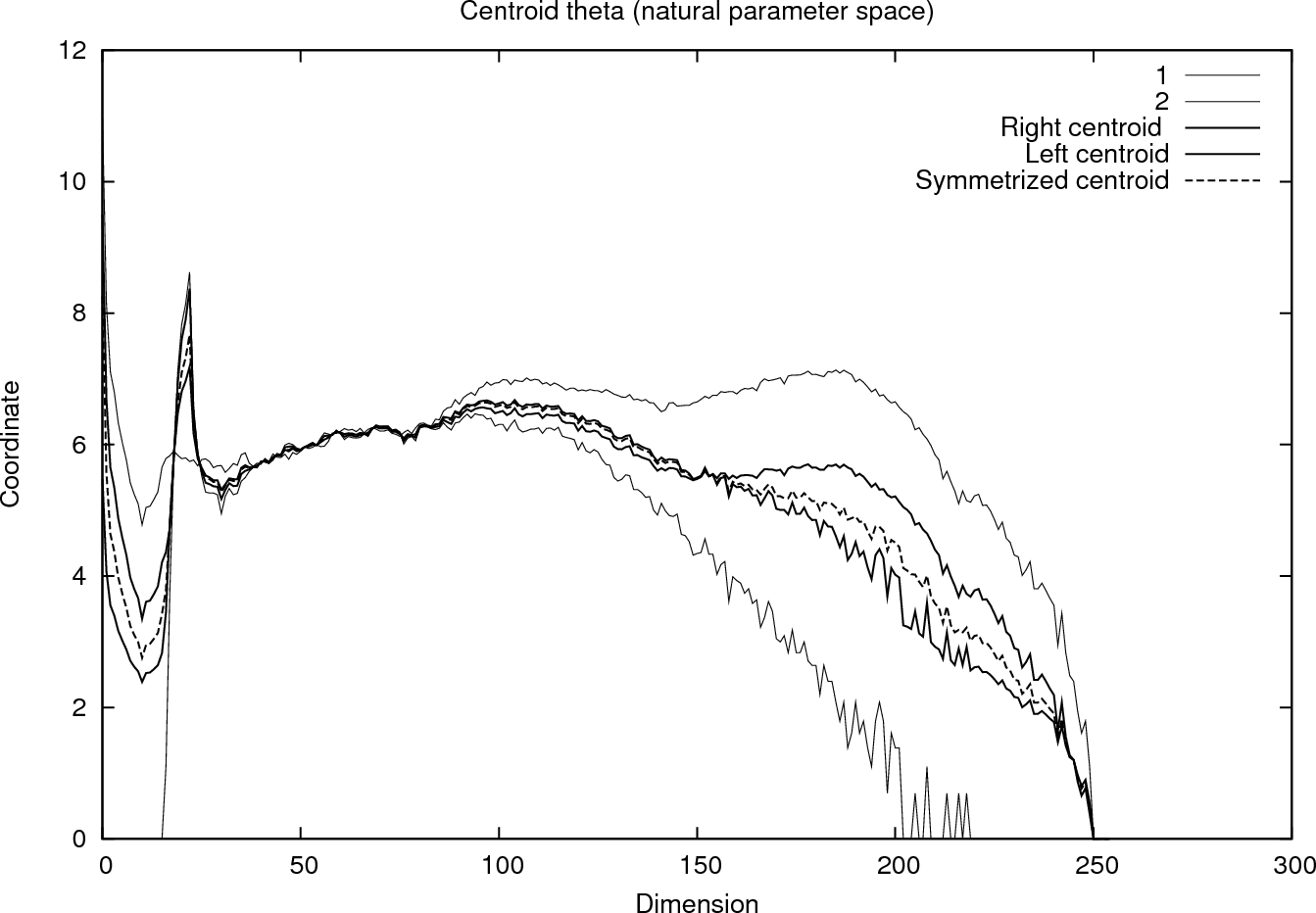}\\
\includegraphics[bb=0 0 1320 938, width=10cm]{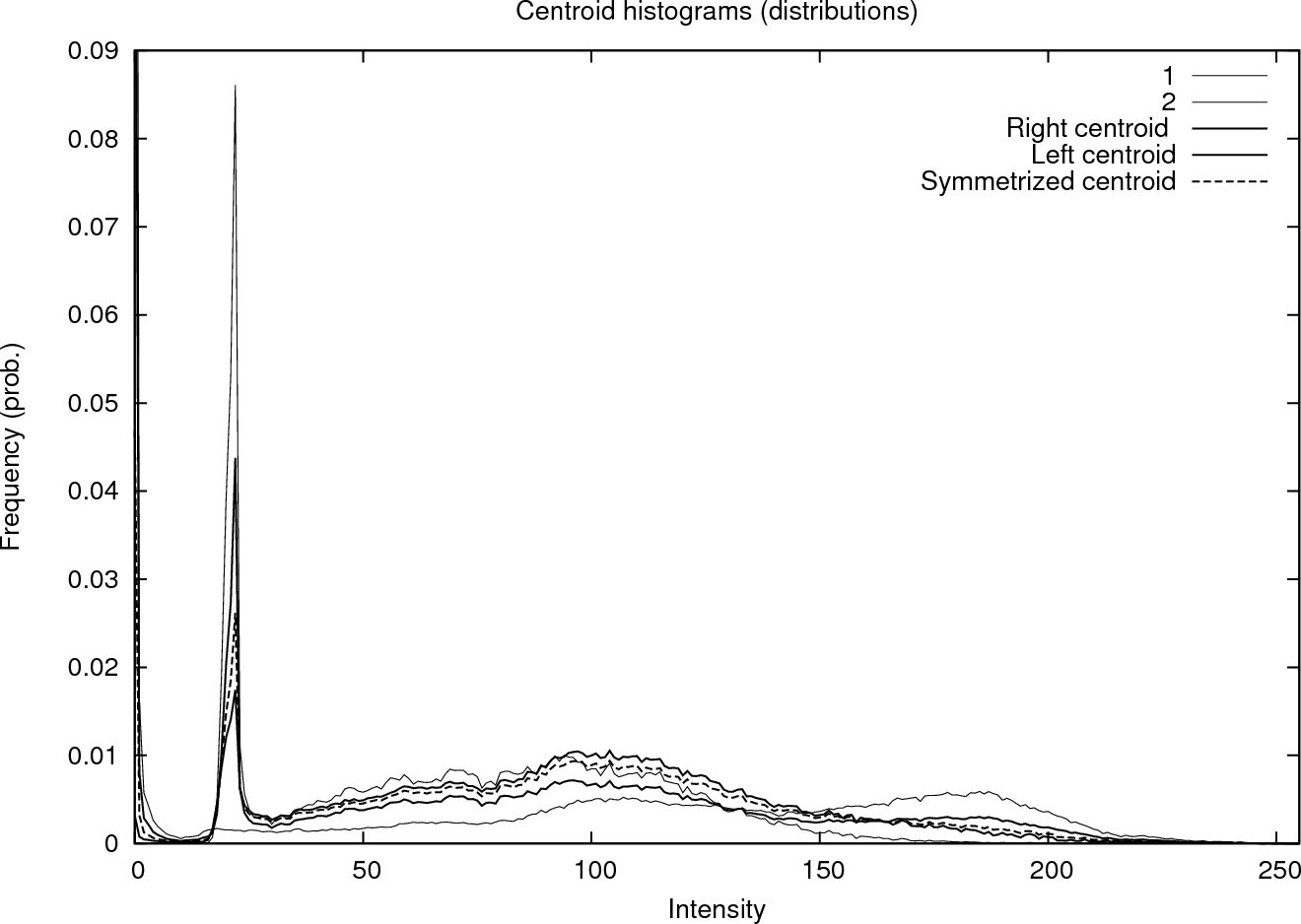}
\end{center}
}

\newpage

\def\pmatf#1#2#3#4{{\left[\begin{array}{cc} #1 & #2 \cr #3 & #4 \end{array}\right] }}

\section{Entropic sided and symmetrized centroids of bivariate normal distributions\label{appendix:normal}}

We report on our implementation for multivariate normal distributions below.
Observe that the right-type Kullback-Leibler centroid is a left-type Bregman centroid for the log normalizer of the exponential family. 
Our method allowed us to verify that the simple generalized $\grad F$-mean formula $c_L^F(\P)=(\grad F)^{-1}(\sum_{i=1}^n \frac{1}{n} \nabla F (p_i))$ coincides with that of the NIPS*06 paper~\cite{gaussianclustering-2006}. 
Furthermore, we would like to stress out that our method extends to {\it arbitrary} entropic centroids of members of the same exponential family.

The figure below plots the entropic right- and left-sided and the symmetrized centroids in red, blue and green respectively for a set that consists of two bivariate normals ($D=\frac{d(d+3)}{2}=5$). The geodesic midpoint interpolant (obtained for $\lambda=\frac{1}{2}$) is very close to the symmetrized centroid, and shown in magenta.

\begin{center}

\begin{tabular}{ll}
\begin{minipage}[c]{0.45\textwidth} 
\includegraphics[bb=0 0 520 546 , width=0.95\textwidth]{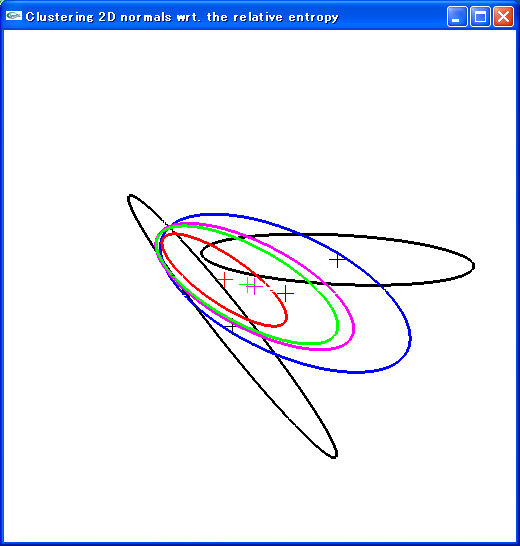}
\end{minipage}
&
\begin{minipage}[c]{0.45\textwidth} 
\tiny
$m_0=(0.34029138065736869,0.26130947813348798)$,\\
$S_0=\pmatf{0.43668091668767117}{-0.42663095837289156}{-0.42663095837289161}{0.63899446830332574)}$\\
$m_1=(0.95591075380718404,0.6544489172032838)$,\\
$S_1=\pmatf{0.79712692342719804}{-0.033060250957646142}{-0.033060250957646142}{0.14609813043797121}$\\
\vskip 0.1cm

$m_R=(0.29050997932657774,0.53527112890397821)$,\\
$S_R=\pmatf{0.33728018979019664}{-0.13844874409795613}{-0.13844874409795613}{0.2321103610207193}$\\
$m_L^F=(0.64810106723227623,0.45787919766838603)$,\\
$S_L^F=\pmatf{0.71165072320677747}{-0.16933954090511438}{-0.16933954090511441}{0.43118595400867693}$\\
$\mathbf{m}^F=(0.42475123207621085,0.5062178606510539)$,\\
$\mathbf{S}^F=\pmatf{0.50780328118070528}{-0.15653432651371618}{-0.15653432651371618}{0.30824860232457035}$\\
$m_{\frac{1}{2}}=(0.46930552327942698,0.49657516328618234)$,\\
$S_{\frac{1}{2}}=\pmatf{0.55643330303588234}{-0.16081280872294987}{-0.1608128087229499}{0.33314553526979185}$.\\
\vskip 0.15cm
Information radius:
\begin{itemize}
\item  right, left: $0.83419372149741644$
\item  {\bf symmetrized}: $0.64099815325721565$
\item  geodesic $\lambda=\frac{1}{2}$: $0.6525069280087431$
\end{itemize}
\end{minipage}
\end{tabular}

\end{center}

We give other pictorial results below for $n=2$ and $n=10$ bivariate normals, respectively.

\begin{center}

\begin{tabular}{ll}
\begin{minipage}[c]{0.45\textwidth} 
\includegraphics[bb=0 0 520 546 , width=0.85\textwidth]{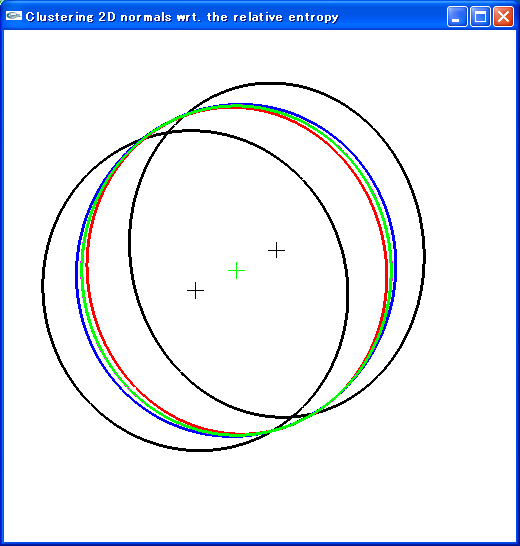}
\end{minipage}
&
\begin{minipage}[c]{0.45\textwidth} 
\includegraphics[bb=0 0 520 546 , width=0.85\textwidth]{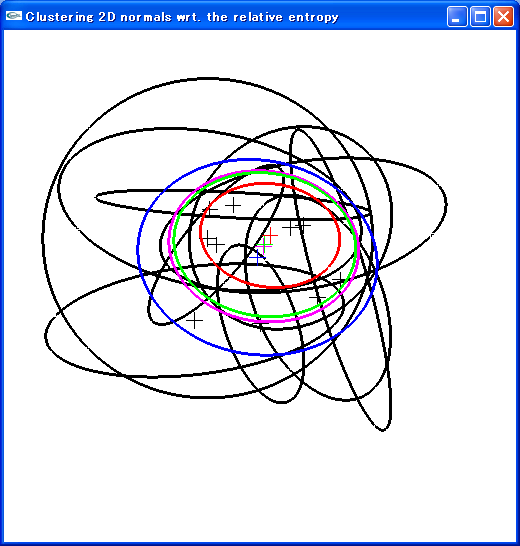}
\end{minipage}
\end{tabular}

\end{center}
\end{document}